\documentclass[conference]{IEEEtran}

\usepackage{cite}
\usepackage{amsmath,amsthm}
\usepackage{verbatim,amsfonts,graphicx,amsmath,amsbsy,amssymb,epsfig,epsf,url}

\newtheorem{theorem}{Theorem}

\newtheorem{LEMM}{Lemma}

\newcommand{\Ib}{{\mathbf I}}

\newcommand{\Xb}{{\mathbf X}}
\newcommand{\Yb}{{\mathbf Y}}
\newcommand{\Zb}{{\mathbf Z}}

\newcommand{\bb}{{\mathbf b}}

\newcommand{\eb}{{\mathbf e}}

\newcommand{\mb}{{\mathbf m}}

\newcommand{\ub}{{\mathbf u}}

\newcommand{\wb}{{\mathbf w}}
\newcommand{\xb}{{\mathbf x}}
\newcommand{\yb}{{\mathbf y}}
\newcommand{\zb}{{\mathbf z}}
\newcommand{\vtheta}{{\boldsymbol{\theta}}}
\newcommand{\vmu}{{\boldsymbol{\mu}}}
\newcommand{\vphi}{{\boldsymbol{\phi}}}
\newcommand{\vGamma}{{\boldsymbol{\Gamma}}}
\newcommand{\vLambda}{{\boldsymbol{\Lambda}}}
\newcommand{\vSigma}{{\boldsymbol{\Sigma}}}
\newcommand{\vPhi}{{\boldsymbol{\Phi}}}
\newcommand{\vDelta}{{\boldsymbol{\Delta}}}
\newcommand{\vXi}{{\boldsymbol{\Xi}}}

\ifCLASSINFOpdf
\else
\fi

\begin{document}
%
\title{Belief Propagation for Joint Sparse Recovery}



%
\author{\IEEEauthorblockN{Jongmin Kim\IEEEauthorrefmark{1},
Woohyuk Chang\IEEEauthorrefmark{2}, 
Bangchul Jung\IEEEauthorrefmark{3},
Dror Baron\IEEEauthorrefmark{4}, and
Jong Chul Ye\IEEEauthorrefmark{1}}
\IEEEauthorblockA{\IEEEauthorrefmark{1}Dept. of Bio and Brain Engineering,
KAIST,
Guseong-dong, Daejon 305-701,  Korea \\ Email: franzkim@gmail.com, jong.ye@kaist.ac.kr}
\IEEEauthorblockA{\IEEEauthorrefmark{2}Research Laboratory of Electronics, MIT, MA, USA,
Email: whchang@mit.edu}
\IEEEauthorblockA{\IEEEauthorrefmark{3}Dept. of Info. and Comm. Engineering, Gyeongsang National Univ. Korea}
\IEEEauthorblockA{\IEEEauthorrefmark{4}Electrical and Computer Engineering Dept., North Carolina State University,  Raleigh, NC, USA}
}


\maketitle

\begin{abstract}
Compressed sensing (CS) demonstrates that
sparse signals can be recovered from underdetermined linear
measurements. We focus on the joint sparse recovery problem where multiple signals share the same common sparse support sets,
and they are measured through the same sensing matrix. 
Leveraging a recent
information theoretic characterization of single
signal CS, we formulate the optimal minimum mean square error
(MMSE) estimation  problem,  and derive a belief
propagation algorithm,  its relaxed version, for the joint sparse recovery problem and an
approximate message passing algorithm.
In addition, using density evolution,  we provide a sufficient condition for exact recovery.
%
\end{abstract}

\IEEEpeerreviewmaketitle

\section{Introduction}
Compressed sensing (CS)~\cite{CandesRUP,DonohoCS}
has revolutionalized sparse signal processing from underdetermined
linear measurements. CS offers a sharp contrast to the
traditional sensing and processing paradigm that first sample the
entire data at the Nyquist rate, only to later throw away most of
the coefficients.
Owing to the potential for reduced measuring rates, CS has
become an active research area in signal processing.

An important area in compressed sensing research is known as
distributed CS~\cite{BaronDCStech}. Distributed CS is based on
the premise that joint sparsity within signal ensembles enables a
further reduction in the number of measurements. Motivated by sensor
networks~\cite{Pottie2000}, preliminary work in distributed
CS~\cite{Duarte2006IPSN,BaronDCStech} showed that the number of
measurements required per sensor must account for the minimal
features unique to that sensor, while features that are common to
multiple sensors are amortized among sensors. 
Distributed CS led to a
proliferation of research on the multiple measurement vector problem
(MMV)~\cite{chen2006trs,kly2010cmusic,lb2010imusic}.
The MMV problem considers the recovery of a set of sparse signal vectors that 
share common non-zero supports through an identical sensing matrix, and ties into several applications of interest,  such as sensor networks, radar, parallel MRI, etc.

In single signal CS, recent results have established the fundamental
performance limits in the presence of
noise~\cite{RFG2009,GuoBaronShamai2009}. For sparse measurement
matrices, belief propagation CS reconstruction~\cite{CSBP2010} is
asymptotically optimal. Based on the revelation that the posteriors
in CS signal estimation are similar in form to outputs of scalar
Gaussian channels~\cite{GuoBaronShamai2009}, additional recent
results~\cite{Rangan2010,donoho2009message} have demonstrated the potential
for faster algorithms for implementing BP.  Another recent
breakthrough is the discovery of an
approximate message passing (AMP) algorithm \cite{donoho2009message},
which was originally derived as a fast approximation of BP, and is strikingly similar
to iterative thresholding~\cite{DaDeDe04} while achieving
theoretically optimal mean square error. 


Leveraging the aforementioned recent progress in the single measurement vector CS,
this paper extends BP to the MMV problem. In particular, we show that BP can be formulated as vector message passing by considering the general signal correlation structures between input vectors. 
Next, a relaxed BP algorithm is derived based
on a Gaussian assumption for the messages, and conditions for the edge independent covariance update are rigorously derived. Finally, we provide 
AMP update rules by further removing the edge dependency of the mean update. As a byproduct, 
 we provide a sufficient condition for exact joint sparse recovery using AMP state evolution.


\section{Signal and Measurement Model}

\bigskip

{\bf Signal model}:\
We consider a model in which an ensemble of~$J$ length-$N$ signals are jointly
sparse as follows.
Our notation for a matrix $\Xb$ uses $\xb^n$ for the $n$th row vector and 
$\xb_j$ as the $j$th column vector.
 
\begin{itemize}
\item [(S1)] Each signal $\{\xb_j\}_{j=1}^J$ belongs to $\mathbb{R}^N$, i.e., $\xb_j\in\mathbb{R}^N$, and for each $n=1,\cdots,N$, the $n$-th component $x_{nj}$ of $\xb_j$ for $j=1,\cdots,J$ is given by $x_{nj}=b_nu_{nj}$.
\item [(S2)] The $N$ random variables $\{b_n\}_{n=1}^N$ representing  support are independent and identically distributed (iid) Bernoulli random variables (RV's) with probability $\epsilon$ of being one and $1-\epsilon$ of being zero.
\item [(S3)] The random vectors $\ub^n\in \mathbb{R}^J$ 
are iid random vectors, which have a multivariate normal distribution $\mathcal{N}_J(\ub;0,\vLambda)$
with zero mean and covariance matrix $\vLambda$, 
    $$f_{\ub}(\ub^n)=\epsilon\mathcal{N}_J((\ub^n)^T;{\bf 0},\vLambda)+(1-\epsilon)\delta((\ub^n)^T)\ .$$ 
\end{itemize}

{\bf Measurement model}:\
For each $\mathbf{x}_j$, a measurement vector $\mathbf{y}_j \in \mathbb{R}^M$ containing $M$
noisy linear measurements is derived by multiplying the signal $\mathbf{x}_j$
by a measurement matrix $\vPhi\in\mathbb{R}^{M\times N} $, and adding noise $\mathbf{z}_j$,
\[
\mathbf{y}_j=\vPhi\mathbf{x}_j+\mathbf{z}_j,\quad j=1,\cdots,J,
\]
where the noise vector $\zb_j$ is i.i.d zero mean Gaussian with noise variance $\sigma^2$.
Following the terminology of sensor array signal processing, $\yb_j$ denotes the $j$-th snapshots, and we refer $J$ as the number of snapshots.

To analyze belief propagation (BP), we consider the factor graph
$G=(V,F,E)$ with variable node $V=[N]=\{1,\cdots,N\}$, 
factor nodes $F=[M]$, and
edges $E\subset\{(m,n):m\in[M],n\in[N]\}$ so that $G$ is a bipartite
graph with $M$ factor nodes and $N$ variable nodes. We let
$E=\{(m,n)\in [M]\times [N]:\Phi_{mn}\neq 0\}$.
We consider a large system limit where $\epsilon$, $\sigma$ and $J$ are constants, but the signal
length $N$ goes to infinity, and the number of measurements $M=M(N)$
also goes to infinity, 
\[
\lim_{N\rightarrow\infty} \frac{M(N)}{N} = \delta,
\]
where $\delta>0$ in problems of practical interest. In this setting,
we let $d$ be a positive integer such that $d<M$ and the elements of $\Phi$ depends on $d$.
Then we let $d\rightarrow\infty$, 
so that we can utilize the
central limit theorem to analyze the large system limit. For the
measurement matrix $\vPhi$ and the edges $E$ of factor graph, we
assume the following:
\begin{itemize}
\item [(M1)] The subgraphs $G_m^i$ and $G_n^i$ of the factor graph within
$2i$ hops from factor node $m$ and variable node $n$ are trees, meaning that
there are no local loops in the graph (for precise definitions, see \cite{Rangan2010}).
\item [(M2)] For all $n\in \{1,\cdots,N\}$,
$$|\{1\leq m\leq M:\Phi_{mn}\neq 0\}|=O(d)$$
and for all $m\in\{1,\cdots,M\}$,
$$|\{1\leq n\leq N:\Phi_{mn}\neq 0\}|=O(d)$$
respectively. Moreover, for any $(m,n)\in E$,
$\Phi_{mn}=O(1/\sqrt{d})$ as $d,N\rightarrow \infty$.
\item [(M3)] For all factor nodes $l$ in $G_m^i$, we have
\begin{eqnarray*}
\lim\limits_{d\rightarrow\infty}\lim\limits_{N\rightarrow\infty}\sum\limits_{n=1}^N (\Phi_{ln})^2=\frac{1}{\delta},~\lim\limits_{d\rightarrow\infty}\lim\limits_{N\rightarrow\infty}\sum\limits_{n=1}^N (\Phi_{ln})^3=0.
\end{eqnarray*}
For all variable nodes $r$ in $G_n^i$, we have
\begin{eqnarray*}
\lim\limits_{d\rightarrow\infty}\lim\limits_{M\rightarrow\infty}\sum\limits_{m=1}^M (\Phi_{mr})^2=1,~\lim\limits_{d\rightarrow\infty}\lim\limits_{M\rightarrow\infty}\sum\limits_{m=1}^M (\Phi_{mr})^3=0.
\end{eqnarray*}
\end{itemize}

\section{Vector Belief Propagation}


Let $\Xb:=[x_{ij}]_{i=1,j=1}^{N~J}$ and $\xb_j=[x_{1j},\cdots,x_{Nj}]^T$,
$g(\Xb)=p(\Yb=\Yb_0|\Xb)$, $h(\bb)=p(\bb)$,  $f_n(\xb^n,b_n)=
p(\xb^n|b_n)$ for $1\leq n\leq N$ and $1\leq j\leq J$. Then, to compute the MMSE estimate of $\xb^n$, we need $p(\xb^n|\Yb=\Yb_0)$, which is given as
\begin{eqnarray*}
&&p(\xb^n|\Yb=\Yb_0)\\
&=&\sum\limits_{\bb\in \{0,1\}^N}\int_{\Xb^{-n}}p(\Xb,\bb|\Yb=\Yb_0)\\
&\propto&\sum\limits_{b_n=0}^1p_n(\xb^n,b_n)p(b_n)\int_{\Xb^{-n}}g(\Xb)\prod\limits_{q\neq n}
\sum\limits_{b_q=0}^1 p_q(\xb^q,b_q)\\
&&\times\sum\limits_{\bb_{-n,q}\in \{0,1\}^{N-2}}p(\bb_{-n}|b_q)
\end{eqnarray*}
where $\Xb^{-n}$ denotes the collection of $\{\xb^j\}_{1\leq j\leq N,~j\neq n}$, $\bb_{-n}$ denotes $\bb$ with the $n$-th element omitted and $\bb_{-n,q}$ denotes $\bb$ with both $n$-th and $q$-th elements omitted. Due to independence,  $p(\bb_{-n}|b_n)=p(\bb_{-n,q}|b_q,b_n)p(b_q)$, we have the following:
\begin{eqnarray*}
&&p(\xb^n|\Yb=\Yb_0)\propto \nu_{f_n\rightarrow \xb^n}(\xb^n)\nu_{g\rightarrow \xb^n}(\xb^n),\\
&&\nu_{f_n\rightarrow \xb^n}:=\sum\limits_{b_n=0}^1 f_n(\xb^n,b_n)p(b_n),\\
&&\nu_{g\rightarrow \xb^n}:=\int_{\Xb^{-n}}g(\Xb)\prod\limits_{q\neq n}\sum\limits_{b_n=0}^1
f_q(\xb^q,b_q)p(b_q) \  .
\end{eqnarray*}

We let $\mu_{A\rightarrow B}(\cdot)$ denote a message passed from node $A$ to its adjacent node $B$ in the factor graph. Extending the sum-product rule of belief propagation to the vector case, the messages can be represented by the following equations where we use the superscript $(i)$ to denote estimated posteriors during iteration $i$:

\begin{eqnarray*}
\nu_{f_n\rightarrow \xb^n}^{(i)}(\xb^n)&\propto&\sum\limits_{b_n=0}^1 f_n(\xb^n,b_n)\nu_{b_n\rightarrow f_n}^{(i)}(b_n)\\
\nu_{g\rightarrow \xb^n}^{(i)}(\xb^n)&\propto&\int_{\Xb^{-n}}g(\Xb)\prod\limits_{q\neq n}
\sum\limits_{b_q=0}^1f_q(\xb^q,b_q)\nu_{b_q\rightarrow f_q}^{(i)}(b_q).
\end{eqnarray*}
%
%

In general, if the measurement matrix is sparse so that the factor graph has local tree-like properties, then
belief propagation produces the true marginal distribution of $\xb^n$ given the observations $\Yb_0$ \cite{GuoWang2008}.
 For dense matrices, belief propagation shows some interesting optimality properties in the large system limit \cite{GuoBaronShamai2009,donoho2009message}. 
 However, the complexity of evaluating marginal distributions grows exponentially in $d$ so that exact belief propagation is not suitable for dense matrices.

\section{Relaxed BP}\label{sec:rbp}

\subsection{Derivation}
Guo and Wang's original work \cite{GuoWang2008} presented an important results that the mean-square optimality of BP could be derived by a significantly simpler algorithm called relaxed BP. Relaxed BP overcomes the limitation of the BP by using a Gaussian approximation of the messages to minimize computation. Therefore, similar to Guo and Wang, we assume $\nu_{g_m\rightarrow \xb^n}^{(i)}(\xb^n)$ to be Gaussian under relaxed BP formulation.

Suppose that the measurement matrix $\vPhi$ satisfies the conditions (M1), (M2) and (M3). We let $\vmu_m^q(i)$ and $\vGamma_m^q(i)$ be the mean and covariance of $(\xb^q)^T$ with pdf $\nu_{\xb^q\rightarrow g_m}^{(i)}$ at the $i$-th iteration, respectively.   
Since we have the following linear relation:
$$\yb^m=\Phi_{mn}\xb^n+\sum\limits_{q\neq n}\Phi_{mq}\xb^q+\zb^m$$
for $m=1,\cdots,M$ and $n=1,\cdots,N$, 
the Gaussian form of the message  $\nu_{g_m\rightarrow \xb^n}^{(i)}$  is represented as
\begin{eqnarray}\label{prod:normal:general}
\nu_{g_m\rightarrow\xb^n}^{(i)}(\xb^n) 
&\propto&\mathcal{N}_J(\Phi_{mn}(\xb^n)^T;\zb_n^m(i),\vSigma_n^m(i)) \ .
\end{eqnarray}
Here, owing to the assumption that the rows of $\Xb$ are independent and $\zb^m$ has zero mean,  we can easily derive: 
\begin{eqnarray}
\zb_n^m(i)&:=&(\yb^m)^T-\sum\limits_{q\neq n}\Phi_{mq}\vmu_m^q(i) \label{eq:zb_n} \\
\vSigma_n^m(i)&:=&\sigma^2I+\sum\limits_{q\neq n}|\Phi_{mq}|^2\vGamma_m^q(i) \label{eq:vSigma} \  ,
\end{eqnarray}
since $\vGamma_m^q(i)$ is the error variance of $\xb^q$ with the pdf $\nu_{\xb^q\rightarrow g_m}^{(i)}$.
Now, we want to identify the pdf of the message $\nu_{\xb^q\rightarrow g_m}^{(i)}$.
Due to the sum-product rule, the message $\mu_{\xb^n\rightarrow g_m}^{(i+1)}(\xb^n)$ is given by
\begin{eqnarray}\label{nu-nm}
\nu_{\xb^n\rightarrow g_m}^{(i+1)}(\xb^n)&\propto&\nu_{f_n\rightarrow\xb^n}(\xb^n)\prod\limits_{q\neq m}
\nu_{g_q\rightarrow\xb^n}^{(i)}(\xb^n).
\end{eqnarray}
We already know that individual messages $\nu_{g_q\rightarrow\xb^n}^{(i)}(\xb^n)$  within the product are Gaussian. Hence, the product is also Gaussian.  
%
Hence, the pdf of the message is given by
\begin{eqnarray}\label{eq-nu-nm-general}
\nu_{\xb^n\rightarrow g_m}^{(i+1)}(\xb^n)&\propto&
[\epsilon \mathcal{N}_J((\xb^n)^T;{\bf{0}},\vLambda)+(1-\epsilon)\delta((\xb^n)^T)] \nonumber\\
&&\times \mathcal{N}_J
((\xb^n)^T;\vtheta_m^n(i),\tilde{\vSigma}_m^n(i))
\end{eqnarray}
where
$$\vtheta_m^n(i)=\left[\sum\limits_{l\neq m}
|\Phi_{ln}|^2(\vSigma_n^l(i))^{-1}\right]^{-1}\sum_{l\neq m} \Phi_{ln}(\vSigma_n^l(i))^{-1}\zb_n^l(i)$$
and
$$\tilde{\vSigma}_m^n(i)=\left[\sum\limits_{l\neq m}
|\Phi_{ln}|^2(\vSigma_n^l(i))^{-1}\right]^{-1},$$
which are calculated  by using the following formula:
\begin{eqnarray*}
\prod\limits_q \mathcal{N}_J(\xb,\mb_q,\vSigma_q)&\propto& \mathcal{N}_J(\xb,\tilde{\mb},\tilde{\vSigma})
\end{eqnarray*}
where
\begin{eqnarray*}
\tilde{\mb}&=&\left[\sum_q \vGamma_q^{-1}\right]^{-1}\sum_q \vGamma_q^{-1}\mb_q,\\
\tilde{\vSigma}&=&\left[\sum_q \vSigma_q^{-1}\right]^{-1}.
\end{eqnarray*}


Using Lemma \ref{lem-mean-var-gen}  in Appendix~A, we have the following  message passing rule for relaxed BP.
%
\begin{eqnarray*}
\label{theta-nm}\vtheta_m^n(i)&=&\left[\sum\limits_{l\neq m}
|\Phi_{ln}|^2(\vSigma_n^l(i))^{-1}\right]^{-1}\sum_{l\neq m} \Phi_{ln}(\vSigma_n^l(i))^{-1}\zb_n^l(i),\\
\label{mu-mn}\vmu_m^n(i)&=&\eta(\vtheta_m^n(i-1);\tilde{\vSigma}_m^n(i-1))\\
\label{gamma-nm}\vGamma_m^n(i)&=&V(\vtheta_m^n(i-1);\tilde{\vSigma}_m^n(i-1))\\
\label{z-mn}\zb_n^m(i)&=&(\yb^m)^T-\sum\limits_{q\neq n}\Phi_{mq}\vmu_m^q(i),\\
\label{sigma-mn}\vSigma_n^m(i)&=&\sigma^2I+\sum\limits_{q\neq m}|\Phi_{mq}|^2\vGamma_m^q(i)
\end{eqnarray*}
where
$\tilde{\vSigma}_n^m=\left[\sum\limits_{l\neq m}|\Phi_{ln}|^2(\vSigma_n^l(i))^{-1}\right]^{-1},$
and
\begin{eqnarray*}
\eta(\vtheta_m^n(i);\tilde{\vSigma}_m^n(i))&=&t_{mn}(i) \wb_{mn}(i) \label{eq:F} \\ 
V(\vtheta_m^n(i);\tilde{\vSigma}_m^n(i))&=&t_{mn}(i)(1 -t_{mn}(i) )\wb_{mn}(i)\wb^H_{mn}(i)
\\&&+t_{mn}(i) \left(\vLambda^{-1}+\left(\tilde{\vSigma}_m^n(i)\right)^{-1}\right)^{-1},
\end{eqnarray*}
with  $$
\wb_{mn}(i) = \left(\vLambda^{-1}+\left(\tilde{\vSigma}_m^n(i)\right)^{-1}\right)^{-1}\left(\tilde{\vSigma}_m^n(i)\right)^{-1}\vtheta_m^n(i),$$
 $$t_{mn}(i)=t(\vtheta_m^n(i);\tilde{\vSigma}_m^n(i)))$$ where
\begin{eqnarray}\label{eq:t}
t {(\vtheta;\vSigma)}=\frac{1}{1+\frac{1-\epsilon}{\epsilon}{|\vGamma+\vSigma|^{\frac{1}{2}}}/{|\vSigma|^{\frac{1}{2}}}
e^{-\frac{1}{2}\vtheta^T(\vLambda^{-1}-(\vSigma+\vGamma)^{-1})\vtheta}
} \  .
\end{eqnarray}
%


The update rule of relaxed BP is still complicated due to edge dependence of the messeges.
In particular, most of the computational overhead comes from the calculation of edge dependent   $\tilde{\vSigma}_n^m$ due to the matrix inversion.  Recall that  $\tilde{\vSigma}_n^m$ denotes the variance of the accumulated error from individual messages $\nu_{\xb^q\rightarrow g_m}^{(i)}, \forall q\neq n$.  Our goal is therefore to derive a edge {\it in}dependent relaxed BP algorithm that removes the dependency of $n,m$ in $\tilde{\vSigma}_n^m$  using the relaxed belief propagation in the large system limit. For this,  we need some extensions of the law of large numbers and the central limit theorem, which are given in \cite{Rangan2010}.
Using these results, we  can now remove the edge dependence of the 
message passing rule for relaxed BP as shown in the following theorem.
\begin{theorem}\label{thm:independency}
Consider the relaxed BP where $\vPhi$ and $(m,n)\in E$ satisfy (M1),
(M2) and (M3) for some fixed iteration number $i\geq 2$. Then as
$N,d\rightarrow\infty$, we have:
\begin{eqnarray}
\label{th-con} \lim\limits_{d,N\rightarrow \infty}\vtheta_m^n(k)&\sim&\Xb+\Zb(k-1)\\
\label{z-con} \lim\limits_{d,N\rightarrow\infty}\zb_n^m(k)&\sim&\mathcal{N}_J(0,\vSigma(k-1))\\
\label{s-con} \lim\limits_{d,N\rightarrow\infty}\vSigma_m^n(k)&=&\vSigma(k):=\sigma^2 I+\frac{1}{\delta}\vGamma(k)
\end{eqnarray}
and
\begin{equation}\label{ts-con}
\lim\limits_{d,N\rightarrow\infty}\tilde{\vSigma}_n^m(k)=\vSigma(k)
\end{equation}
for $k\leq i$, where $\Xb$ and $\Zb(k-1)$ has pdf $f_{\xb}:=\epsilon\mathcal{N}_J(\xb;0,\vLambda)+(1-\epsilon)\delta(\xb)$ and $\mathcal{N}_J(0,\vSigma(k-1))$ and
$$\vGamma(k):=E[V(\Xb+\Zb(k-1),\vGamma(k-1))],$$
and
$\vmu_m^n(1)=\hat{\xb}:=E(\Xb)$, $\vGamma_m^n(1)=\vGamma(1):={\rm
Cov}(\Xb)$ for all $m$ and $n$.
\end{theorem}
\begin{proof}
See Appendix~B.
\end{proof}

When the measurement matrix is sparse, in the
large system limit, if the average degree $d$ grows as
$o(M^{1/(4k)})$, then there is a so-called asymptotic cycle-free
property\cite{GuoWang2008}. That is, the possibility of existence of
a cycle of length shorter than $k$ approaches zero. Hence, in the
large system limit,  the assumption (M1) in the above theorem is 
asymptotically correct if
$d=o(M^{1/(4k)})$.
Under this condition,  the edge
independence of $\vSigma(i)$  proved in the above theorem lead us to  replace the message passing rule
for relaxed BP as 
\begin{eqnarray}
\vmu_m^n(i)&=&\eta(\vtheta_m^n(i-1);\vSigma(i-1))\label{eq:mu-mn}\\
\vGamma_m^n(i)&=&V(\vtheta_m^n(i-1);\vSigma(i-1))\\
\vtheta_m^n(i)&=&\sum_{l\neq m} \Phi_{ln}\zb_n^l(i), \label{eq:theta-nm} \\
\zb_n^m(i)&=&(\yb^m)^T-\sum\limits_{q\neq n}\Phi_{mq}\vmu_m^q(i), \label{eq:z-mn}\\
\vSigma(i)&=&\sigma^2I+\frac{1}{\delta}\vGamma(i),
\end{eqnarray}
where
$$\vGamma(i)=\frac{1}{|E|}\sum\limits_{(m,n)\in E}V(\vtheta_m^n(i);\vSigma(i)).$$

\section{Approximate Message Passing for MMV}\label{sec:amp}

Recently, Donoho, Maleki and Montanari \cite{donoho2009message} developed the approximate message passing (AMP) for single measurement vector(SMV) problem $\yb=\vPhi \xb$, which shows significant advantages over the conventional iterative thresholding algorithm, while achieving similar performance to basis pursuit. The AMP was developed within the belief propagation framework.  
 In order to execute the belief propagation (or relaxed belief propagation), we must keep track of $2MN$ messages, but in applying AMP, we just need to keep track of $M+N$ messages so that AMP reduces computation. AMP is more suitable to large-scale applications, whereas basis pursuit often demands too much time. 


To derive AMP for MMV, we let
\begin{eqnarray}
\label{app-theta}\vtheta_m^n(i)&=&\vtheta^n(i)+\delta\vtheta_m^n(i)+O(1/d),\\
\label{app-mu}\vmu_m^n(i)&=&\vmu^n(i)+\delta\vmu_m^n(i)+O(1/d),\\
\label{app-z}\zb_n^m(i)&=&\zb^m(i)+\delta\zb_n^m(i)+O(1/d).
\end{eqnarray}
Substituting (\ref{app-z}) into (\ref{eq:theta-nm}),  
 (\ref{app-mu}) into (\ref{eq:z-mn}), and (\ref{app-theta}) into (\ref{eq:mu-mn}), respectively, we have the following results.
 
 \begin{theorem}
For the given signal model (S1)-(S3) and the measurement model  (M1)-(M3), the  approximate message passing algorithm for multiple measurement vectors is given by
\begin{eqnarray}
\label{mu-n-a}\vmu^n(i+1)&=&\eta(\vtheta^n(i);\vSigma(i))\\
\label{theta-n-a}\vtheta^n(i)&=&\sum\limits_{l=1}^M\Phi_{ln}\zb^l(i)+\vmu^n(i)\\
\label{z-m-a}\zb^m(i+1)&=&(\yb^m)^T-\sum\limits_{q=1}^N
\Phi_{mq}\vmu^q(i+1)\\
&&+\zb^m(i)\sum\limits_{q=1}^N  \eta'(\vtheta^q(i);\vSigma(i))
\Phi_{mq}^2 \end{eqnarray}
where
\begin{eqnarray}
\label{v-a}\vGamma(i+1)&=&\frac{1}{N}\sum\limits_{n=1}^N
\left[(t_n(i)-t_n^2(i))\wb_n(i)\wb_n^H(i) \right. \nonumber \\
&&\qquad\qquad\left.+t_n(i)(\vLambda^{-1}+(\vSigma(i))^{-1})^{-1}\right] \nonumber \\
\vSigma(i)&=&\sigma^2I+\frac{1}{\delta}\vGamma(i) 
\end{eqnarray}
and $t_n(i) = t(\vtheta^n(i);\vSigma)$ and $\wb_n(i)=(\vLambda^{-1}+\vSigma(i)^{-1})^{-1}(\vSigma(i))^{-1}\vtheta^n(i)$,
and $\eta'(\vtheta^q(i);\vSigma(i))$ denotes the derivatives of $\eta(\vtheta^q(i);\vSigma(i))$ with respect to $\vtheta^q(i)$, respetively.
\end{theorem}
\begin{proof}
See Appendix~C.
\end{proof}


%
%

\section{Case Study: Uncorrelated Snapshots}\label{amp-uncorr}

The  AMP update rule can be further simplified when
the input source vectors $\{\xb_j\}_{j=1}^J$  are uncorrelated to
each other. This scenario is the most optimistic in estimating the
sparse support since ${\rm rank}(\Xb)$ determines the upper bound of
maximal sparsity \cite{kly2010cmusic}. 

 More
specifically, consider the signal model (S1)-(S3) with $\vLambda=\Ib$, where $\Ib$ denotes the identity matrix. 
In this setting,   $\vSigma(i) = c(i)\Ib$ and $\vGamma(i)=\gamma(i)\Ib$  so
that we have
\begin{eqnarray*}
\vmu_{n}(i+1)=\eta(\vtheta^n(i),c(i)\Ib)= t_n(i) \frac{\vtheta^n(i)
}{1+c(i)}
\end{eqnarray*}
where the shrinkage operator $t_n(i) \equiv t(\vtheta^n(i),c(i)\Ib)$ is given by
\begin{eqnarray*}
t_n(i)=
\frac{1}{1+\frac{1-\epsilon}{\epsilon}(1+c^{-1}(i))^{\frac{J}{2}}\exp\{-\frac{\|\vtheta^n(i)\|^2}{2c(i)(1+c(i))}\}} \  .
%
\end{eqnarray*}
Fig.~\ref{fig:shrinkage} plots the shrinkage operator output with respect to the  normalized input value, $\|\vtheta^n(i)\|^2/J$, for various $J$ parameters when $c(i)=0.1$. As $J$ increases, it clearly exhibits a hard-thresholding behaviour with the threshold value of
$ {c(i)(1+c(i))\log(1+c^{-1}(i))}$ (see Appendix~D for proof).
\begin{figure}[!htbp]
\begin{center}
\epsfig{figure=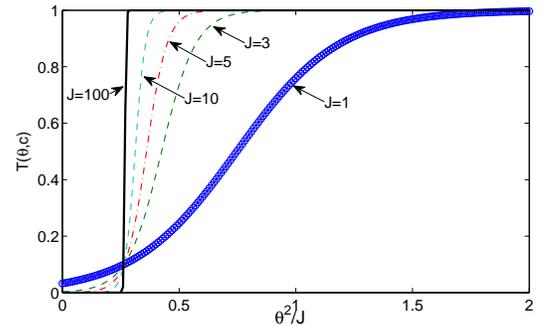,width=8cm}
\caption{Shrinkage operator for various number of snapshots for  $c(i)=0.1$ and $\epsilon=0.1$.}
\label{fig:shrinkage}
\end{center}
\end{figure}
Thanks to the hard-thresholding behavior, the AMP update rule can be further simplified.
First, we can easily see that $t(i)(1-t(i))\rightarrow 0$.
Therefore, 
we have
\begin{eqnarray}
\gamma(i+1) &=&  \frac{c(i)}{1+c(i)}\epsilon(i) 
\end{eqnarray}
where 
\begin{eqnarray}
c(i)&=&\sigma^2+\frac{1}{\delta}\frac{c(i-1)}{1+c(i-1)}\epsilon(i-1)  \\
\epsilon(i)  &=  & \frac{1}{N}\sum_{n=1}^N t_n(i)  \  ,
\end{eqnarray}
which denotes the ratio of the row whose $l_2$ norm exceeds the threshold.
In a large system limit as $N \rightarrow \infty$,  Appendix~E shows that $\epsilon(i)=E[t_n(i)]=\epsilon$.
Therefore, the corresponding  state evolution is
\begin{equation}\label{se-limit-j}
c(i+1)=\sigma^2+\frac{\epsilon}{\delta}\frac{c(i)}{1+c(i)}.
\end{equation}
The following theorem provides an important observation for the convergence of the state evolution.
\begin{theorem}\label{thm:sampling_rate}
	In noiseless case, $c(i)$ converges to 0 regardless of the initial condition if and only if $\epsilon\leq \delta$. 
\end{theorem}
\begin{proof}
See Appendix~F.
\end{proof}	
Theorem~\ref{thm:sampling_rate} informs us that  the minimum undersampling ratio for AMP convergence approaches the sparsity rate $\epsilon$ as the number of snapshots increases. Considering  the existing results \cite{BaronDCStech} stating that $\epsilon$ is the minimum sampling rate we can achieves,  AMP provides a computationally efficient framework to achieve the optimality.

\section{Numerical Results}

Here, the experimental parameters are as following: $M=50, N=100$,  $J=3$, $\epsilon=0.1$ and $d=20$.
The sparse sensing matrix $\vPhi$ that satisfy (M1)-(M3) are generated by first drawing elements of $\{-1,1\}$ with equal probability, retaining the values with the probability of $d/M$, and scaling by $1/\sqrt{d}$. 

\begin{figure}[htbp]
 \centerline{
     \epsfig{figure=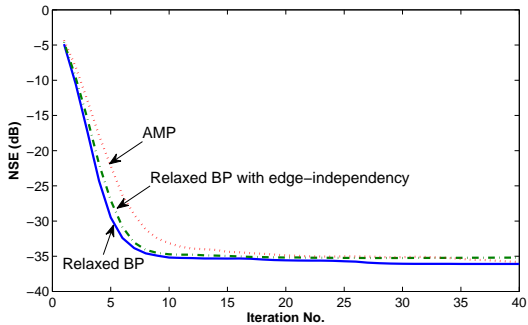,width=8cm}}
  \caption{Convergence of relaxed BP, relaxed BP with edge-independency and AMP.}
 \label{fig:convergence}
\end{figure}
Fig.~\ref{fig:convergence} illustrates the normalized squared-error (NSE) for relaxed BP, relaxed BP with edge-indepedence, and AMP,
when the signal correlation matrix $\vLambda= \Ib$ and SNR=30dB.
The results in Fig.~\ref{fig:convergence}  clearly demonstrate that all algorithms converges to the nearly equivalent MSE value.

\section{Conclusion}\label{sec:con}
%
 We showed that a vector form of message passing is appropriate to describe belief propagation in MMV problem.  
Then, we adopted the idea of Guo and Wang to approximate the message as Gaussian pdf and provided a relaxed BP algorithm by only passing mean and covariances.  It turns out that the resulting relaxed BP has an interesting shrinkage operator within the update as  a function of norm of the signal row vector. 
To reduce the computational overhead, we derived a rigorous condition for an edge independent covariance update for the relaxed BP. Finally, we derived the AMP algorithm that totally removes edge dependence even in mean update, which has complexity comparable to other iterative thresholding algorithms.  
Furthermore,  using state evolution,  we derived  a sufficient condition for joint sparse recovery, which showed that the AMP achieves the optimality as the number of snapshot increases.

%
%
%
%

\section*{Acknowledgment}

This work was supported by the Korea Science and Engineering Foundation (KOSEF) grant funded by the Korea government (MEST) (No.2010-0000855).


\bibliographystyle{IEEEtran}
\bibliography{IEEEabrv,cites,bispl2010}

\newpage

\section*{Appendix  A}

For the calculation of mean and variance of $\nu_{\xb^q\rightarrow g_m}^{(i)}$ , we need the following lemma.
\begin{LEMM}\label{lem-mean-var-gen}
Suppose that a random variable $\xb\in\mathbb{R}^J$ has pdf $f(\xb)$ as
\begin{eqnarray*}
f(\xb)&\propto&\left[
\epsilon\frac{\exp{(-\frac{1}{2}\xb^T\vLambda^{-1}\xb)}}{(2\pi)^{\frac{J}{2}}|\vLambda|^{\frac{1}{2}}}
+(1-\epsilon)\delta(\xb)\right]\\
&& \times \frac{\exp{(-\frac{1}{2}(\xb-\vtheta)^T\vSigma^{-1}(\xb-\vtheta))}}{(2\pi)^{\frac{J}{2}}|\
\vSigma|^{\frac{1}{2}}},
\end{eqnarray*}
for some $\vtheta\in\mathbb{R}^J$, then we have the followings:
\begin{eqnarray*}
E(\xb)&=&t(\vtheta;\vSigma)\vphi \\ 
{\rm Cov}(\xb)&=&
t(\vtheta;\vSigma)(\vphi\vphi^T+(\vLambda^{-1}+\vSigma^{-1})^{-1})-t(\vtheta;\vSigma)^2\vphi\vphi^T.
%
%
\end{eqnarray*}
where
\begin{eqnarray*}
\vphi  &=& (\vLambda^{-1}+\vSigma^{-1})^{-1}\vSigma^{-1}\vtheta,\\
t {(\vtheta;\vSigma)}&=&\frac{1}{1+\frac{1-\epsilon}{\epsilon}{|\vLambda+\vSigma|^{\frac{1}{2}}}/{|\vSigma|^{\frac{1}{2}}}
e^{-\frac{1}{2}\vtheta^T(\vSigma^{-1}-(\vSigma+\vLambda)^{-1})\vtheta}
} \  .
\end{eqnarray*}
\end{LEMM}

\begin{proof}
Let
\begin{eqnarray*}
f(x)&=&k\left[
\epsilon\frac{\exp{(-\frac{1}{2}\xb^T\vLambda^{-1}\xb)}}{(2\pi)^{\frac{J}{2}}|\vLambda|^{\frac{1}{2}}}
+(1-\epsilon)\delta(\xb)\right]\\
&&\times \frac{\exp{(-\frac{1}{2}(\xb-\vtheta)^T\vSigma^{-1}(\xb-\vtheta))}}{(2\pi)^{\frac{J}{2}}|\vSigma|^{\frac{1}{2}}}\\
&=&k\left[\epsilon\frac{\exp{(-\frac{1}{2}\xb^T\vLambda^{-1}\xb-\frac{1}{2}(\xb-\vtheta)^T\vSigma^{-1}(\xb-\vtheta))}}{(2\pi)^J
|\vLambda|^{\frac{1}{2}}|\vSigma|^{\frac{1}{2}}} \right. \\
&& \left.
+(1-\epsilon)\frac{\exp{(-\frac{1}{2}\vtheta^T\vSigma^{-1}\vtheta)}}{(2\pi)^{\frac{J}{2}}|\vSigma|^{\frac{1}{2}}}\delta(\xb)\right].
\end{eqnarray*}
Since
\begin{eqnarray}
&&\xb^T\vLambda^{-1}\xb+(\xb-\vtheta)^T\vSigma^{-1}(\xb-\vtheta)\nonumber\\
&=&(\xb-\vphi)^T\vXi^{-1}(\xb-\vphi)
+\vtheta^T\vDelta^{-1}\vtheta,\label{eq:var}
\end{eqnarray}
where
\begin{eqnarray*}
\vXi = (\vLambda^{-1}+\vSigma^{-1})^{-1}, &
\vphi=\vXi \vSigma^{-1}\vtheta ,&
\vDelta = \vLambda+\vSigma \  .
\end{eqnarray*}
By plugging in Eq.~\eqref{eq:var} into the probability density function $f(\xb)$ and integrating out, we have
%
\begin{eqnarray*}
k&=&\left[\epsilon\frac{1}{(2\pi)^{\frac{J}{2}}|\Delta|^{\frac{1}{2}}}\exp{\left(-\frac{1}{2}
\vtheta^T\Delta^{-1}\vtheta\right)} \right. \\
&&\left. +(1-\epsilon)\frac{1}{(2\pi)^{\frac{J}{2}}|\vSigma|^{\frac{1}{2}}}\exp{\left(-\frac{1}{2}
\vtheta^T\vSigma^{-1}\vtheta\right)}\right]^{-1}.
\end{eqnarray*}
The resulting pdf $f(x)$ is then given by
\begin{eqnarray*}
f(\xb)&=&
\frac{\frac{e^{-\frac{1}{2}(\xb-\vphi)^T\vXi^{-1}
(\xb-\vphi)}}{(2\pi)^{\frac{J}{2}}|\vXi|^{\frac{1}{2}}}
+ \frac{1-\epsilon}{\epsilon}\frac{e^{-\frac{1}{2}(\vtheta^T[\vSigma^{-1}-\vDelta^{-1}]\vtheta}}{|\vSigma|^{\frac{1}{2}}/|\vDelta|^{\frac{1}{2}}}\delta(\xb)}
{1+\frac{1-\epsilon}{\epsilon}{|\vLambda+\vSigma|^{\frac{1}{2}}}/{|\vSigma|^{\frac{1}{2}}}
e^{-\frac{1}{2}\vtheta^T(\vSigma^{-1}-(\vSigma+\vLambda)^{-1})\vtheta}}
\end{eqnarray*}
whose  mean is 
\begin{eqnarray*}
E(\xb)=t(\vtheta;\vSigma)\vphi
\end{eqnarray*}
and the covariance is 
\begin{eqnarray*}
&&{\rm Cov}(\xb\xb^T)\\&=&E(\xb\xb^T)-E(\xb)E(\xb^T)\\
&=&T(\vtheta;\vSigma)(\vphi\vphi^T+(\vLambda^{-1}+\vSigma^{-1})^{-1})-t(\vtheta;\vSigma)^2\vphi\vphi^T .
\end{eqnarray*}

\end{proof}

\section*{Appendix B}

We need the following two theorems to prove the claim.
\begin{theorem}[Law of Large numbers \cite{Rangan2010}]\label{lln}
For each $N$ and $d$, let $m(N,d)=O(d)$ and $\xb_{N,i}^d\in\mathbb{R}^J$, $i=1,\cdots,m$ be a set of independent random variables satisfying
$$\lim\limits_{d\rightarrow\infty}\lim\limits_{N\rightarrow\infty}\xb_{N,i}^d\sim \Xb$$ where $\Xb$ denotes a random variable with  a pdf  $f_{\Xb}(\xb)$, which denotes the distribution of limiting random vector $\xb$ and $i=i(N,d)=\{1,\cdots,m\}$ is any deterministic sequence.
Here, $X\sim Y$ denotes that two random vectors $X$ and $Y$ have the same distributions. Let $a_{N,i}^d$ be a set of non-negative deterministic constants such that $b_{N,i}^d=O(1/\sqrt{d})$ and
$$\lim\limits_{d\rightarrow\infty}\lim\limits_{N\rightarrow\infty}\sum\limits_{i=1}^ma_{N,i}^d=1.$$
Then
$$\lim\limits_{d\rightarrow\infty}\lim\limits_{N\rightarrow\infty}\sum\limits_{i=1}^ma_{N,i}^d\xb_{N,i}^d\sim E(\Xb).$$
\end{theorem}
\begin{theorem}[Central Limit Theorem \cite{Rangan2010}]\label{clt}
Let $\xb_{N,i}$ be as in Theorem \ref{lln} such that for any deterministic sequence of indices $i=i(N,d)\in \{1,\cdots,m\}$, we have the limit
$$\lim\limits_{d\rightarrow\infty}\lim\limits_{N\rightarrow\infty}\sqrt{d}|E(\xb_{N,i}^d)-E(\Xb)|=0.$$
Also suppose that $a_{N,i}^d$ be a set of non-negative deterministic constants such that $a_{N,i}^d=O(1/\sqrt{d})$ and
$$\lim\limits_{d\rightarrow\infty}\lim\limits_{N\rightarrow\infty}\sum\limits_{i=1}^m|a_{N,i}^d|^2=1~{\rm and}~\lim\limits_{d\rightarrow\infty}\lim\limits_{N\rightarrow\infty}\sum\limits_{i=1}^m(a_{N,i}^d)^3=0.$$
Then
$$\lim\limits_{d\rightarrow\infty}\lim\limits_{N\rightarrow\infty}\sum\limits_{i=1}^m a_{N,i}^d(\xb_{n,i}^d-E(\Xb))\sim
\mathcal{N}(0,{\rm var}(\Xb)).$$
\end{theorem}

{\em Proof of Theorem~2}:
 First, by applying (M2) and (M3), we can easily see that
(\ref{ts-con}) holds for $k=r$ if (\ref{s-con}) holds for $k=r$.
Hence we will show that (\ref{z-con}), (\ref{s-con}) and
(\ref{ts-con}) holds for $k=1$, and the claim holds for any $k\leq
i$ by induction. By (\ref{sigma-mn}), with
$\vGamma_m^q(1)=\vGamma(1)={\rm Cov}(\Xb)$,
$$\vSigma_n^m(1)=\sigma^2I+\sum\limits_{q\neq n}|\Phi_{mq}|^2\vGamma(1).$$
By the assumption (M2) and (M3), we can easily see that
$\vSigma_n^m(1)\rightarrow \sigma^2I+(1/\delta)\vGamma(1)$ as
$d,N\rightarrow\infty$. Also, by (\ref{z-mn}) and
$\vmu_m^n(1)=E(\Xb)=\hat{\xb}$, we have
\begin{eqnarray*}
&&\notag \zb_n^m(1)\nonumber\\
&=&(\yb^m)^T-\sum\limits_{q\in N(m)\setminus \{n\}}\Phi_{mq}\vmu_m^q(1)\nonumber \\
\notag&=&\sum\limits_{q\in N(m)}\Phi_{mq}(\xb^q)^T+(\zb^m)^T-\sum\limits_{q\in N(m)\setminus \{n\}}\Phi_{mq}\vmu_m^q(1)\\
\label{zmn2}&=&\Phi_{mn}(\xb^n)^T+(\zb^m)^T+\sum\limits_{q\in
N(m)\setminus \{n\}}\Phi_{mq}((\xb^q)^T-\hat{\xb}). \nonumber
\end{eqnarray*}
By the assumption (M1), the terms in the sum of (\ref{zmn2}) are
independent and (M2) makes the first term disappear so that by the
modified central limit theorem and the condition (M2),
\begin{eqnarray*}
\zb_n^m(1)&\sim& \mathcal{N}_J\left(0,\sigma^2I+\frac{1}{\delta}{\rm Cov}(\Xb-\hat{\xb})\right)\\
&=&\mathcal{N}_J(0,\sigma^2I+\frac{1}{\delta}\vGamma(1)).
\end{eqnarray*}

Now, we show that if (\ref{th-con}) holds for $k=r$, then
(\ref{z-con}) and (\ref{s-con}) holds for $k=r$. By the definition
of $\vmu_m^n(r)$ and $\vGamma_m^n(r)$, we have
$$\lim\limits_{d,N\rightarrow\infty}\vmu_m^n(r)\sim \eta(\Xb+\Zb(r-1);\vSigma(r-1))$$ and
$$\lim\limits_{d,N\rightarrow\infty}\vGamma_m^n(r)\sim V(\Xb+\Zb(r-1);\vSigma(r-1))$$
where $\Xb$ and $\Zb(r-1)$ has pdf
$f_{\xb}:=\epsilon\mathcal{N}_J(\xb;0,\vLambda)+(1-\epsilon)\delta(\xb)$
and $\mathcal{N}_J(\zb(r-1);0,\vSigma(r-1))$, respectively. Note
that $$\vSigma_n^m(r)=\sigma^2I+\sum\limits_{q\neq
n}|\Phi_{mq}|^2\vGamma_m^q(r).$$ By the assumption that $G_{mn}(r)$
is a tree, the terms in the above summation are statistically
independent. Since $\lim_{d,N\rightarrow\infty}\vGamma_m^n(r)\sim
V(\Xb+\Zb(r-1);\vSigma(r-1))$, by the law of large numbers in
Theorem \ref{lln},
\begin{eqnarray*}
\lim\limits_{d,N\rightarrow\infty}\vSigma_n^m(r)&\sim& \sigma^2I+\frac{1}{\delta}E[V(\Xb+\Zb(r-1);\vSigma(r-1))] \\
&=&\sigma^2I+\frac{1}{\delta}\vGamma(r).
\end{eqnarray*}
Next, we consider $\zb_n^m(r)$. Note that
\begin{eqnarray}
\notag \zb_n^m(r)&=&(\yb^m)^T-\sum\limits_{q\neq n}\Phi_{mq}\vmu_m^q(r)\\
\label{z-mn-conv}&=&\zb^m+\Phi_{mn}\vmu_m^n(r)+\sum\limits_{q\neq
n}\Phi_{mq}((\xb^q)^T-\vmu_m^q(r)). \nonumber
\end{eqnarray}
By the assumption (M2),
$\Phi_{mn}\vmu_m^n(r)=O(1/\sqrt{d})\rightarrow 0$ as $d,N\rightarrow
\infty$. Furthermore, we have
$$\lim_{d,N\rightarrow\infty}((\xb^q)^T-\vmu_m^q(r))\sim \Xb-\eta(\Xb+\Zb(r-1);\vSigma(r-1))$$ 
so that
\begin{eqnarray*}
\lim_{d,N\rightarrow\infty}{\rm Cov}[(\xb^q)^T-\vmu_m^q(r)]
\end{eqnarray*}
\begin{eqnarray*}
&=&
{\rm Cov}[(\Xb-\eta(\Xb+\Zb(r-1);\vSigma(r-1)))]\\
&=&E[V(\Xb+\Zb(r-1);\vSigma(r-1))]=\vGamma(r).
\end{eqnarray*}
By using the central limit theorem on (\ref{z-mn-conv}), we have
$\lim_{d,N\rightarrow\infty}\zb_n^m(r+1)\sim
\mathcal{N}_J(0,\sigma^2I+(1/\delta)\vGamma(r))$.

\bigskip

Finally, we show that if (\ref{z-con}) and (\ref{s-con}) holds for
$k=r$, (\ref{th-con}) holds for $k=r+1$. By the induction
hypothesis, we have $\vSigma_n^m(r)\rightarrow\vSigma(r)$ for all
$(m,n)\in E$. By the assumption (M3), we have
\begin{eqnarray*}\label{theta-mn-conv}
&&\vtheta_n^m(r)\\
&=&\left[\sum\limits_{l\neq m}
|\Phi_{ln}|^2(\vSigma_n^l(r))^{-1}\right]^{-1}\sum_{l\neq m}
\Phi_{ln}(\vSigma_n^l(r))^{-1}\zb_n^l(r)\nonumber\\
&\longrightarrow&
\sum\limits_{l\neq m}\Phi_{ln}\zb_n^l(r) \nonumber
\end{eqnarray*}
as $d,N\rightarrow\infty.$ Letting
$\eb^l(r)=\zb_n^l(r)-\Phi_{ln}(\xb^n)^T$, by (M2) and the
assumption, we have ${\rm Cov}[\eb^l(r)]\rightarrow \vSigma(r-1)$ as
$d,N\rightarrow \infty$. Then we have
\begin{eqnarray*}
\lim\limits_{d,N\rightarrow\infty}\vtheta_m^n(r)&=&\lim\limits_{d,N\rightarrow\infty}[\sum\limits_{l\neq m}\Phi_{ln}\eb^l(r)+\sum\limits_{l\neq m}\Phi_{ln}^2(\xb^n)^T]\\
&=&(\xb^n)^T+\sum\limits_{l\neq m}\Phi_{ln}\eb^l(r)
\end{eqnarray*}
by using (M2) and (M3). By the central limit theorem,
$\vtheta_m^n(r)\sim \Xb+\Zb(r-1)$.

\section*{Appendix C}
Substituting (\ref{app-z}) into (\ref{eq:theta-nm}),  we have
\begin{eqnarray*}
\vtheta_m^n(i)&=&\sum\limits_{l\neq m}\Phi_{ln}[\zb^l(i)+\delta\zb_n^l(i)]
=\sum\limits_{l=1}^M \Phi_{ln}[\zb^l(i)+\delta\zb_n^l(i)]\\&&-\Phi_{mn}\zb^m(i)+O(1/d)
\end{eqnarray*}
so that the followings hold:
\begin{eqnarray}
\label{theta-n}\vtheta^n(i)&=&\sum\limits_{l=1}^M\Phi_{ln}\zb_n^l(i),\\
\label{dtheta-nm}\delta\vtheta_m^n(i)&=&-\Phi_{mn}\zb^m(i).
\end{eqnarray}
Similarly, substituting (\ref{app-mu}) into (\ref{eq:z-mn}), we have
\begin{eqnarray*}
\zb_n^m(i)&=&(\yb^m)^T-\sum\limits_{q\neq n}\Phi_{mq}[\vmu^q(i)+\delta\vmu_m^q(i)] \\
&=&(\yb^m)^T-\sum\limits_{q=1}^N \Phi_{mq}\vmu^q(i)+\Phi_{mn}\vmu^n(i)+O(1/d)
\end{eqnarray*}
so that we have:
\begin{eqnarray}
\label{z-m}\zb^m(i)&=&(\yb^m)^T-\sum\limits_{q=1}^N \Phi_{mq}\vmu^q(i)\\
\label{dz-m}\delta\zb_n^m(i)&=&\Phi_{mn}\vmu^n(i).
\end{eqnarray}
Finally, substituting (\ref{app-theta}) into (\ref{eq:mu-mn}), we have the following Taylor series expansion
\begin{eqnarray*}
\vmu_m^n(i+1)&=&\eta(\vtheta_m^n(i);\vSigma(i))\vtheta_m^n(i)\\
&\cong&
\eta(\vtheta^n(i);\vSigma(i))\vtheta^n(i)-\eta'(\vtheta^n(i);\vSigma(i))\Phi_{mn}\zb^m(i)
\end{eqnarray*}
so that
\begin{eqnarray}
\label{mu-n}\vmu^n(i+1)&=&\eta(\vtheta^n(i);\vSigma(i))\\
\label{dmu-n}\delta\vmu_m^n(i+1)&=&-\eta'(\vtheta^n(i);\vSigma(i)))\Phi_{mn}\zb^m(i).
\label{eq:DF}
\end{eqnarray}
Hence $\vtheta^n(i)$ is updated according to
\begin{eqnarray}\label{theta-n-update}
\vtheta^n(i)&=&\sum\limits_{l=1}^M\Phi_{ln}\zb_n^l(i)
=\sum\limits_{l=1}^M\Phi_{ln}[\zb^l(i)+\Phi_{ln}\vmu^n(i)]\notag\\
&=&\sum\limits_{l=1}^M\Phi_{ln}\zb^l(i)+\vmu^n(i) \ ,
\end{eqnarray}
in the large system limit by (\ref{theta-n}) and (\ref{dz-m}) and (M3).
 Also, $\zb^{m}(i+1)$ is updated according to
\begin{eqnarray*}\label{z-m-update}
&&\zb^m(i+1)\\
&=&(\yb^m)^T-\sum\limits_{q=1}^N\Phi_{mq}\vmu_m^q(i+1) \nonumber\\
&=&(\yb^m)^T-\sum\limits_{q=1}^N
\Phi_{mq}[\vmu^q(i+1)-\eta'(\vtheta^q(i);\vSigma(i))\Phi_{mq}\zb^m(i)]\notag\\
&=&(\yb^m)^T-\sum\limits_{q=1}^N
\Phi_{mq}\vmu^q(i+1)+\sum\limits_{q=1}^N
\eta'(\vtheta^q(i);c(i))\Phi_{mq}^2\zb^m(i). \nonumber\\
\end{eqnarray*}
Here,  $\eta'(\vtheta;\vSigma)$  in \eqref{eq:DF} can be calculated
using the first order derivative of \eqref{eq:F} with respect to
$\vtheta$:
\begin{eqnarray}
  \eta'(\vtheta;\vSigma) &=& T(\vtheta;\vSigma)(\vLambda^{-1}+\vSigma^{-1})^{-1}\vSigma^{-1} \nonumber\\
  & +& (\vLambda^{-1}+\vSigma^{-1})^{-1}\vSigma^{-1}\vtheta
  t'(\vtheta;\vSigma),
\end{eqnarray}
where $t'(\vtheta;\vSigma)$ denotes the derivative of
$t(\vtheta;\vSigma)$ with respect to $\vtheta$: 
\begin{eqnarray*}\label{eq:DT}
  t'(\vtheta;\vSigma) 
   &=&  t^2(\vtheta;\vSigma) \frac{1-\epsilon}{\epsilon}\vtheta^T\left(\vSigma^{-1}-(\vSigma+\vLambda)^{-1}\right) \\
   && \times \frac{|\vLambda+\vSigma|^\frac{1}{2}}{|\vSigma|^{\frac{1}{2}}}
   e^{-\frac{1}{2}\vtheta^T\left(\vSigma^{-1}-(\vSigma+\vLambda)^{-1}\right)\vtheta}  \ .
  \end{eqnarray*}

\section*{Appendix  D }

%
%
Note that
\begin{eqnarray*}
\lim_{J\rightarrow \infty} (1+c^{-1})^{\frac{J}{2}}e^{-\frac{\theta^2}{2c(1+c)}}
&=& \lim_{J\rightarrow \infty}e^{\frac{J}{2}\left(\log(1+c^{-1})-\frac{\theta^2}{Jc(1+c)}\right)} 
\end{eqnarray*}
This value becomes 0 when $\frac{\theta^2}{J} > {c(1+c)\log(1+c^{-1})}$; 1 when $\frac{\theta^2}{J} ={c(1+c)\log(1+c^{-1})}$, and $\infty$ otherwise.
Therefore, due to the definition of the shrinkage operator, this concludes the proof.

\section*{Appendix E}

Using Eq.~\eqref{th-con} in Theorem~\ref{thm:independency}, in a large system limit, we have
\begin{eqnarray}
\vtheta^n(i)&\sim&\Xb+\Zb(i-1)
 \end{eqnarray}
where $\Xb$ and $\Zb(k-1)$ has pdf $f_{\xb}:=\epsilon\mathcal{N}_J(\xb;0,\vLambda)+(1-\epsilon)\delta(\xb)$ and $\mathcal{N}_J(0,\vSigma(i-1))$.
Since the two RV's $\Xb$ and $\Zb(i-1)$ are independent, the corresponding pdf can be therefore derived by convolving the two pdfs, providing us
\begin{eqnarray}\label{eq:f}
f\left(\vtheta^n(i)\right) =\epsilon
\frac{ e^{-\frac{\|\vtheta^n(i)\|^2}{2(1+c(i))}}}{(2\pi)^{\frac{J}{2}}{(1+c(i))^{\frac{1}{2}}}}
 +(1-\epsilon)\frac{ e^{-\frac{\|\vtheta^n(i)\|^2}{2c(i)}}}{(2\pi)^{\frac{J}{2}}{c(i)^{\frac{1}{2}}}}
\end{eqnarray}
which can be derived using the similar techniques used in Lemma~\ref{lem-mean-var-gen}.
As $N\rightarrow \infty$, $\sum_{n=1}^N t_n(\vtheta^n(i))/N \rightarrow E\left[t_n(\vtheta^n(i))\right]$. 
Using \eqref{eq:t} and \eqref{eq:f},  we can easily see that $t_n(\vtheta^n(i)) f(\vtheta^n(i)) = \epsilon
{ e^{-\frac{\|\vtheta^n(i)\|^2}{2(1+c(i))}}}/{(2\pi)^{\frac{J}{2}}{(1+c(i))^{\frac{1}{2}}}}$, so we have
\begin{eqnarray*}
E\left[ t_n(\vtheta^n(i))\right] &= &\int t_n(\vtheta^n(i)) f(\vtheta^n(i)) d\vtheta^n(i) \\
&=&  \epsilon \int \frac{ e^{-\frac{\|\vtheta^n(i)\|^2}{2(1+c(i))}}}{(2\pi)^{\frac{J}{2}}{(1+c(i))^{\frac{1}{2}}}} d\vtheta^n(i) \\
&=& \epsilon \ .
\end{eqnarray*}
This concludes the proof.

\section*{Appendix F}

Let us first characterize the behavior at the fixed point $c(i)\rightarrow x$ of the state evolution.
\begin{eqnarray*}
x = \sigma^2 + \frac{\epsilon}{\delta}\frac{x}{1+x}
\end{eqnarray*}
For the noiseless case $\sigma^2=0$,  the fixed points corresponds the intersection of $y=x$ and $y=\frac{\epsilon}{\delta}x/(1+x)$ for $x\geq 0$.
We can easily see that one of the intersection is $x=0$ and the other depends on the slope of $y=\frac{\epsilon}{\delta}x/(1+x)$ at $x=0$. Since the slope is $\epsilon/\delta$, we can easily see that there exist no other intersections other than $x=0$ when the slope is less than or equal to one, i.e. $\epsilon/\delta \leq 1$.
This is the optimal scenario since the resulting error becomes zero regardless of $c(1)$.
Next, to complete the proof, we need to show that the fixed point iteration Eq.~\eqref{se-limit-j} converges. This can be readily shown since $c(i+1)=\frac{\epsilon}{\delta}c(i)/(1+c(i))\leq c(i)$  for $\epsilon\leq \delta$ for all $i\geq 1$. Since the sequence $c(i)$ is monotone decreasing and there exist a fixed solution $c^*=0$, the algorithm converges from any initialization. This concludes the proof.

\end{document}